\documentclass[sigconf]{acmart}

\AtBeginDocument{%
  \providecommand\BibTeX{{%
    \normalfont B\kern-0.5em{\scshape i\kern-0.25em b}\kern-0.8em\TeX}}}

\setcopyright{none}

\renewcommand\footnotetextcopyrightpermission[1]{} 
\usepackage{bm}
\usepackage{algorithm, pseudocode}

\def\K{\ensuremath{\mathbb{K}}}
\def\Kbar {\ensuremath{\overline{\mathbb{K}}}}

\DeclareBoldMathCommand{\c}{c}
\DeclareBoldMathCommand{\a}{a}
\DeclareBoldMathCommand{\f}{f}
\DeclareBoldMathCommand{\g}{g}
\DeclareBoldMathCommand{\h}{P}
\DeclareBoldMathCommand{\x}{x}
\DeclareBoldMathCommand{\z}{z}
\DeclareBoldMathCommand{\v}{v}
\DeclareBoldMathCommand{\0}{0}
\DeclareBoldMathCommand{\u}{u}
\DeclareBoldMathCommand{\e}{u}
\DeclareBoldMathCommand{\R}{\mathcal{R}}
\DeclareBoldMathCommand{\S}{\mathcal{S}}

\def\ZZ {\ensuremath{\mathbb{Z}}}

\def\b_eta{\mbox{\boldmath$\eta$}}
\def\softO{\ensuremath{{O}{\,\tilde{ }\,}}}
\DeclareBoldMathCommand{\balpha}{\alpha}
\def\jac{\ensuremath{{\rm Jac}}}
\def\diag{\ensuremath{\mathrm{diag}}}

\DeclareBoldMathCommand{\H}{H}
\DeclareBoldMathCommand{\G}{G}
\DeclareBoldMathCommand{\F}{F}
\DeclareBoldMathCommand{\P}{P}
\DeclareBoldMathCommand{\Q}{Q}
\DeclareBoldMathCommand{\calG}{\mathcal{G}}

\title{Computing Polynomial Representation in Subrings of Multivariate Polynomial Rings}
\author{Thi Xuan Vu}
\orcid{0000-0002-2285-7801}
\affiliation{%
  \institution{Univ. Lille, CNRS, Centrale Lille, UMR 9189 CRIStAL}
  \postcode{F-59000}
  \city{Lille}
  \country{France}
}
\email{thi-xuan.vu@univ-lille.fr}
\renewcommand{\shortauthors}{Thi Xuan Vu}

\begin{document}

  \begin{abstract}
  Let $\R = \mathbb{K}[x_1, \dots, x_n]$ be a multivariate polynomial ring over a field $\mathbb{K}$ of characteristic 0. Consider $n$ algebraically independent elements $g_1, \dots, g_n$ in $\R$. Let $\S$ denote the subring of $\R$ generated by $g_1, \dots, g_n$, and let $h$ be an element of $\S$. Then, there exists a unique element ${f} \in \K[u_1, \dots, u_n]$ such that $h = f(g_1, \dots, g_n)$. 

In this paper, we provide an algorithm for computing ${f}$, given $h$ and $g_1, \dots, g_n$. The complexity of our algorithm is linear in the size of the input, $h$ and $g_1, \dots, g_n$, and polynomial in $n$ when the degree of $f$ is fixed. Previous works are mostly known when $f$ is a symmetric polynomial and $g_1, \dots, g_n$ are elementary symmetric, homogeneous symmetric, or power symmetric polynomials.

\end{abstract}


\settopmatter{printfolios=true}
\maketitle

\section{Introduction}
\subsubsection*{\bf Problem statement.}
Let $\R = \K[x_1, \dots, x_n]$ denote the ring of multivariate polynomials over a field $\K$, and consider $n$ algebraically independent elements $g_1, \dots, g_n$ in $\R$. We let $\S = \K[g_1, \dots, g_n]$ denote the subring of $\R$ comprising all polynomial expressions in $g_1, \dots, g_n$ with coefficients in $\K$. 

This implies that for any element $h \in \S$, there exists a unique polynomial $f \in \K[u_1, \dots, u_n]$, where $(u_1, \dots, u_n)$ are new variables, such that
\[
h = f(g_1, \dots, g_n).
\]

Since $\S$ is closed under addition, multiplication, and contains the constants, it forms a subring of $\R$. The problem we consider is the computation of $f$ given $h$ and $g_1, \dots, g_n$. For simplicity, we focus here on a single polynomial $h$. However, the results in this paper can be extended to any set $H = (h_1, \dots, h_m)$ of polynomials in $\S$ by applying our algorithm to each $h_i$ individually.

\subsubsection*{\bf Motivation.} Our main motivations of this paper come from solving some structured multivariate polynomial systems and related problems. 
\paragraph{Decomposing Multivariate Polynomials.}  
Given a set of polynomials \( H = (h_1, \dots, h_m) \) over \( \mathbb{K} \), the problem of finding sets of polynomials \( F = (f_1, \dots, f_m) \) and \( G = (g_1, \dots, g_n) \), all in \( \mathbb{K} \), such that  
\[
H = (h_1, \dots, h_m) = (f_1(g_1, \dots, g_n), \dots, f_m(g_1, \dots, g_n))
\]  
is a classical problem in computer algebra.  

In the univariate case, this type of decomposition is a standard functionality implemented in several computer algebra systems. For example, the function \texttt{compoly} in Maple is specifically designed for such purposes. In the multivariate case, there have been numerous contributions \cite{gathen2003multivariate, ding1999cryptanalysis, faugere2006cryptanalysis, ye2001decomposing, faugere2009efficient}, with one of the most recent results provided in \cite{faugere2009efficient}, focusing on polynomials of the same degree.  

A common feature of all known techniques for decomposition is that they divide the problem into two main steps: first, compute candidates for \( g_1, \dots, g_n \), and then recover \( f_1, \dots, f_m \) from this knowledge. Consequently, one of the motivations for this paper is to make progress toward an algorithmic solution for decomposing any multivariate polynomials.  

It is worth noting that, in this context, the degrees of the polynomials \( f_i \) are bounded by the degrees of the original polynomials \( h_i \). Finally, we remark that the functional decomposition problem also has applications in cryptography, particularly in the cryptanalysis of 2R-schemes (see, for instance, \cite{faugere2009efficient} and \cite{patarin1997asymmetric}) and in systems of differential equations, particularly in the angle of model reduction (see
e.g. \cite{feret2009internal, feret2012lumpability, hubert2013scaling,
  cardelli2017erode, cardelli2017maximal} and   
\cite{ovchinnikov2021clue, jimenez2022computing,
  demin2023exact}).

\paragraph{Solving Invariant Polynomial Systems.}  
The computation of \( f \) also plays a critical role in the study of invariant polynomials under the action of finite groups.

Let \( \calG \) be a finite matrix group acting on \( \K[x_1, \dots, x_n] \). The invariant ring of the group \( \calG \) is the ring \( \R^\calG \) of all polynomials invariant under the action of \( \calG \). Then, there exist homogeneous, algebraically independent polynomials \( G = (g_1, \dots, g_n) \) and homogeneous invariants \( \sigma = (\sigma_1, \dots, \sigma_r) \) such that
\[
\R^\calG = \bigoplus_{j=1}^r \K[g_1, \dots, g_n] \, \sigma_j.
\]
Here, \( (g_1, \dots, g_n) \) and \( (\sigma_1, \dots, \sigma_r) \) are called the primary and secondary invariants, respectively. In other words, any \( h \in \R^\calG \) can be uniquely written as
\begin{equation} \label{eq:second-first}
h = \sum_{j=1}^r f_j(g_1, \dots, g_n) \, \sigma_j,
\end{equation}
for some \( f_j \in \K[u_1, \dots, u_n] \).

When \( \calG \) is a finite pseudo-reflection group, \( \R^\calG \) is generated by a finite set of algebraically independent polynomials \( G = (g_1, \dots, g_n) \), known as the basic invariants. That is,
\[
\R^\calG = \K[g_1, \dots, g_n],
\]
or equivalently, any  polynomial \( h \in \R^\calG \) can be uniquely expressed as \( f(g_1, \dots, g_n) \), where \( f \) is in \( \mathbb{K}[u_1, \dots, u_n] \) (see Section \ref{subsec:invariant} for more details).

Our question is motivated by polynomial system solving using algorithms based on \textit{geometric resolution} \cite{giusti1997lower, giusti1998straight, giusti1995polynomial, heintz2000deformation, giusti2001grobner, lecerf2003computing}. The complexity of such algorithms depends on two factors: (i) the complexity of evaluating the input system and (ii) the geometric properties of the system. If the system is invariant under the action of some pseudo-reflection group, a standard approach \cite{gatermann1996semi, Colin97, worfolk1994zeros} is to rewrite it in terms of invariants and then solve the system in the new variables.

To be more precise, let \( H = (h_1, \dots, h_m) \) be a set of polynomials invariant under the action of \( \calG \). Suppose \( F = (f_1, \dots, f_m) \) are polynomials in \( \mathbb{K}[u_1, \dots, u_n] \) such that, instead of solving the system \( h_1 = \cdots = h_m = 0 \) directly, we solve two simpler systems:
\begin{equation} \label{eq:decom}
G = 0 \quad \text{and} \quad F - \balpha = 0,
\end{equation}
where \( \balpha = (\alpha_1, \dots, \alpha_n) \) is a root of \( G = 0 \). This approach allows us to leverage the structure of the invariant ring to simplify the solution of polynomial systems with invariants.

In this paper, we provide an answer to point (i) above, specifically addressing the complexity of evaluating the system in the new variables \( (u_1, \dots, u_n) \).

In addition, the representation \( f \) of \( h \) has significant implications for simplifying computations in invariant theory. By reducing the problem of working with \( f \) to the smaller ring \( \mathbb{K}[u_1, \dots, u_n] \), we can leverage the structure of the basic invariants of \( \calG \) to perform symbolic computations more efficiently (see for examples \cite{FLSSV2021, labahn2023faster, riener2024connectivity, vu2022computing} and references therein). Moreover, understanding the degrees of \( f \) relative to \( h \) and the basic invariants of \( \calG \) provides insight into the algebraic and geometric properties of the invariant ring. However, it is worth noting that while the output \( f \) will be more structured, solving the new system could still be more costly. We consider the study of solving invariant systems using this approach as one of the subjects for future work.

\subsubsection*{\bf Prior Works} 
Note that determining $f$ knowing $h$ and $g_1, \dots, g_n$ is a subfield membership problem (\cite{gathen2003multivariate, sweedler1993using}). This is a difficult problem in general. 

When \( h \) is invariant under the action of the symmetric group and \( (g_1, \dots, g_n) \) are elementary symmetric polynomials, Gaudry, Schost, and Thiéry proved in \cite{gaudry2006evaluation} that the complexity of evaluating \( f \) is bounded by \( \delta(n) L_2 + 2 \), where \( L_2 \) is the complexity of evaluating \( h \), and \( \delta(n) \leq 4^n (n!)^2 \). The authors also mentioned that they do not know whether the factor \( \delta(n) \), which grows polynomially with \( n! \), is optimal. Our result in this paper addresses a more general question, where the factor is the cost of multivariate series multiplication at some precision, rather than \( \delta(n) \).

In the context of symmetric polynomials, where \( h \) and \( (g_1, \dots, g_n) \) are elementary symmetric polynomials, Bläser and Jindal showed that the complexity of evaluating \( f \) is bounded by \( \softO(d^2 L_2 + d^2 n^2) \), where \( d \) is the degree of \( h \). Their algorithm is faster than ours in this paper; however, we address a more general question, not limited to (elementary) symmetric polynomials. It is feasible to have a faster algorithm in this case due to the fact that \( (x_1, \dots, x_n) \) are solutions of the following univariate polynomial:  
\[
P(T) = T^n - \eta_1(\mathbf{x}) T^{n-1} + \cdots + (-1)^{n-1} \eta_{n-1} T + (-1)^n \eta_n,
\]
where \( \eta_k(\mathbf{x}) \) is the \( k \)-th elementary symmetric polynomial. However, this property does not hold for any set of algebraically independent polynomials \( g_1, \dots, g_n \).

Recently, Chaugule et al. \cite[Theorem 4.16]{chaugule2023schur} extended the result of Bläser-Jindal \cite{BlaserJindal18} to other bases, such as the homogeneous symmetric basis and the power symmetric basis, while maintaining the same bound \( \softO(d^2 L_2 + d^2 n^2) \) for the complexity of evaluating \( f \).

In a more general context, when \( h \) is invariant under the action of any finite group \( \calG \), Dahan, Schost, and Wu \cite[Theorem 1]{dahan2009evaluation} proved that the cost of rewriting \( h \) as in \eqref{eq:second-first} is bounded by  
\[
O\big(n^5\delta + n \delta^6 + L_1n^2\delta^4 + L_2 \delta^3\big),
\]  
where \( \delta = \deg(g_1) \cdots \deg(g_n) \) and \( L_1 \) is the complexity of evaluating \( G \) and \( \sigma \). This cost is high (as also mentioned in their paper), especially when the degrees of \( g_i \) are large. Surprisingly, the degree of \( h \) does not appear in the bound.

In general, for any polynomials \(h\) and \(g_1, \dots, g_n\), it is well known that Gröbner bases can be used to compute \(f\) from \(h\) and the \(g_i\)'s. However, the complexity of this process has not been fully analyzed.  Precisely, in the polynomial ring $\K[x_1, \dots, x_n, u_1, \dots, u_n]$, we fix a monomial order $\succ$ where any monomial involving one of $(x_1, \dots, x_n)$ is greater than all monomials in $\K[u_1, \dots, u_n]$. 
Let $\mathcal{B}$ be a {G}r{\"o}bner basis with respect to the order $\succ$ of the ideal $\langle g_1-u_1, \dots, g_n-u_n \rangle \subset \K[x_1, \dots, x_n, u_1, \dots, u_n]$. Then, the polynomial $f$ can be obtained as the remainder of $h$ on division by $\mathcal{B}$.  We refer the reader to \cite[Proposition~4 - Section~1 - Chapter~7]{CLO07} for a detailed description of this procedure.

Another straightforward strategy is to use linear algebra by exploiting the symbolic equality $h = f(g_1, \dots, g_n)$. By comparing the coefficients on both sides of this equality, we obtain a linear system of $\binom{n+\Delta}{n}$ equations in $\binom{n+\Delta}{n}$ unknown coefficients of $f$. The complexity of solving for $f$ using this strategy is $\tilde{O}(\binom{n+\Delta}{n}^\omega)$ operations in $\mathbb{K}$, where $ \Delta$ is a degree bound of $f$ and $\omega$ is  the complexity exponent of linear algebra operations, (for some algorithms, such as Gauss algorithm, $\omega = 3$).

\subsubsection*{\bf Some notations.}  
Let $h$ be a polynomial of degree $d$ in $n$ variables. The usual encoding to represent $h$ is an array of $\binom{n+d}{n}$ coefficients. In contrast, the sparse encoding only represents the non-zero coefficient-exponent tuples. In this paper, we use another way of representing polynomials, which is called the straight-line program encoding, or equivalently, algebraic circuits. 

A \textit{straight-line program} (SLP) computing a set of polynomials $h_1, \dots, h_m$ in $\K[x_1, \dots, x_n]$ is a circuit 
\[
\gamma = (\gamma_{-n+1}, \dots, \gamma_0, \gamma_1, \dots, \gamma_L),
\]
where all $h_i$ belong to \(\gamma\), $\gamma_{-n+1}:=x_1, \dots, \gamma_0 := x_n$, and for $k > 0$, $\gamma_k$ is one of the following forms: 

\begin{itemize}
    \item $\gamma_k = a * \gamma_i$ or $\gamma_k = \gamma_i * \gamma_j$,
\end{itemize}
where $* \in \{+, -, \times\}$, $a \in \K$, and $i, j < k$. Then $L$ is the length of $\gamma$.
Obviously, some polynomials admit short straight-line programs even if they contain many monomials. For example, the polynomial \( (x+1)^k \) has \( k+1 \) monomials, but it can be computed using a straight-line program of length logarithmic in \( k \).

The idea of using straight-line programs first appeared in the probabilistic testing of polynomial identities. In computer algebraic applications, this encoding was first used with the elimination of one-variable problems \cite{heintz1981absolute, kaltofen1988greatest, kaltofen1989factorization}. Our motivation for using this representation in this paper comes from polynomial system-solving algorithms using Newton-Hensel lifting techniques, as initiated in \cite{giusti1997lower, giusti1998straight, giusti1995polynomial} (see also Section \ref{sec:lifting}).  

Note that the use of SLP  representation is not a limitation, as it encompasses both dense and sparse representations. Specifically, any polynomial of degree $d$ in $n$ variables can be computed by an SLP of length at most $3\binom{n+d}{n}$. This follows from the fact that there are $\binom{n+d}{n}$ monomials of degree at most $d$ in $n$ variables. Computing each monomial with its corresponding coefficient and summing them requires at most $2\binom{n+d}{n}$ arithmetic operations. Thus, the total length of the SLP is bounded by $3\binom{n+d}{n}$. For sparse representations, where a polynomial has $N$ nonzero coefficients, a SLP of length $L = O(Nd)$ suffices to compute it.

In what follows, we use $\softO(\cdot)$ to
indicate that polylogarithmic factors are omitted, that is, $a$ is in
$\softO(b)$ if there exists a constant $k$ such that $a$ is $O(b \,
\log^k(b))$ (see~\cite[Chapter 25.7]{Gat03} for technical
details). The smallest integer larger or equal to $a$ is written as
$\lceil a \rceil$.

For positive integers $\Delta$ and $n$, let ${\mathcal{M}}_b(\Delta)$ denote the cost of multiplying univariate polynomials of degree $\Delta$ in terms of operations in the base ring $\K$, and let ${\mathcal{M}}(\Delta, n)$ denote the cost of multiplying $n$-variate power series to total degree $\Delta$. The quantity ${\mathcal{M}}_b(\Delta)$ can be taken in $\tilde{O}(\Delta) = O(\Delta \log \Delta \log \log \Delta)$ using the algorithms of Schönhage and Strassen~\cite{schonhage1971schnelle}, Schönhage~\cite{schonhage1977schnelle}, and Cantor and Kaltofen~\cite{cantor1991fast}. Over a finite field with $q$ elements, under the assumption that there exists a Linnik constant, one can take ${\mathcal{M}}_b(\Delta) = O(n \log q \log(n \log q))$~\cite{harvey2022polynomial}.

The quantity ${\mathcal{M}}(\Delta, n)$ can be bounded above by ${\mathcal{M}}_b((2\Delta+1)^n)$ via Kronecker substitution (see~\cite{Kronecker82} and~\cite[Chapter 8.4]{Gat03}). If the field $\K$ has characteristic zero, then ${\mathcal{M}}(\Delta, n)$ is in $\tilde{O}({\mathcal{M}}_b(\binom{\Delta+n}{n}))$; that is, linear in the size of the dense representation of the series, up to logarithmic factors (see~\cite{lecerf2003fast}).

If two $n$-variate power series are given at precision $\Delta$ by a straight-line program (SLP) of length $L$, then ${\mathcal{M}}(\Delta, n) = O(\Delta^2 L)$. Indeed, there exists a SLP of length $L+1$ to compute the product of the $\Delta$-truncations of the two power series. Moreover, by a result of Krick and Pardo~\cite{krick1996computational}, there exists a SLP of length $(\Delta + 1)^2 (L+1)$ to compute all the homogeneous components of this product — that is, all terms of degrees up to $\Delta$.

\subsubsection*{\bf Main Result}  
Our main result in this paper is the following.
\begin{theorem} \label{main:thm}
Let \(\K\) be a field of characteristic zero, and let \( g_1, \dots, \) \(g_n \) be \( n \) algebraically independent elements in \( \R = \K[x_1, \dots, x_n] \). Let \( \S \) denote the subring \( \K[g_1, \dots, g_n] \) of \( \R \) generated by \( g_1, \dots, g_n \). Then, for any polynomial \(h\) in \(\S\), there exists a randomized algorithm that, given \( (g_1, \dots, g_n) \) and \( h \), returns \( f \) in \( \K[u_1, \dots, u_n] \) such that  
  \[
  f(g_1, \dots, g_n) = h,
  \]
  using  
  \[
  \softO\big((nL_1 + n^4 + L_2) \mathcal{M}(\Delta, n) \big)
  \]  
  operations in \( \K \), where \( L_1 \) and \( L_2 \) are the lengths of the straight-line programs computing \( (g_1, \dots, g_n) \) and \( h \), respectively, and \( \Delta \) is a bound on the degree of \( f \).
\end{theorem}

In particular, when \( h \) is invariant under the action of a finite pseudo-reflection group (e.g., symmetric groups \( S_n \); Weyl groups of Lie algebras such as dihedral groups \( D_n \), signed symmetric groups \( B_n \), etc.), the degree bound \( \Delta \) is bounded by the degree \( d \) of \( h \). Then, we have \( \mathcal{M}(\Delta, n) = \mathcal{M}(d, n) \) (see Theorem \ref{thm:spec}). Note that while our complexity bound does not depend on the degrees of \( g_i \)'s as in \cite{dahan2009evaluation}, it depends on the degree of \( h \) (or more generally, that of \( f \)).

\subsubsection*{\bf Structure of the Paper}  
In the next section, we present our main tool for deriving the main algorithm, which is given in Section~\ref{sec:main}. Additionally, an illustrative example demonstrating the algorithm is provided in Appendix~\ref{sec:ex}.  
As mentioned above, the complexity of our main algorithm depends on the degree bound of the representation polynomial \( f \). Therefore, Section~\ref{subsec:invariant} is dedicated to analyzing such degree bounds, with a specialization  on polynomials invariant under the action of a finite pseudo-reflection group.  
The paper concludes with a summary of our results and a discussion on topics for future research. 


\section{Newton-Hensel lifting}
\label{sec:lifting}

Lifting techniques are classical methods that can be found, for example, in~\cite{heintz2000deformation, giusti2001grobner, schost2003computing} (see also the references therein). 

In this section, the notations $\x$ and $\e$ represent the sets of variables $(x_1, \dots, x_n)$ and $(u_1, \dots, u_m)$, respectively. The ideal $\langle \e \rangle$ in the polynomial ring $\K[\e]$ is generated by $u_1, \dots, u_m$. 
For a positive integer $\delta$, let  
\[
\langle \e \rangle^{\delta} := \langle u_1^{t_1} \cdots u_m^{t_m} \mid t_1 + \cdots + t_m \geq \delta+1 \rangle \subset \K[\e]
\]  
denote the ideal in $\K[\e]$ generated by all monomials of degree at least $\delta+1$. If $p$ is a polynomial in $\K[\e]$ or a power series in $\K[[\e]]$, then $p \mod \langle \e \rangle^\delta$ refers to the part of $p$ containing terms of degree up to $\delta$.
The rest of this section is devoted to proving the following Proposition. 

\begin{proposition} \label{prop:newton}
Let $\h = (p_1, \dots, p_n)$ be a sequence of polynomials in $\K[\x, \e]$, and let $\balpha = (\alpha_1, \dots, \alpha_m) \in \Kbar{}^m$ be a solution of $P(\x, 0, \dots, 0) = 0$. Assume that the Jacobian matrix of $\h(\x, 0, \dots, 0)$ with respect to $\x$ is full rank at $\balpha$. 

Then there exists a unique vector of power series $\v = (v_1, \dots, v_n)$ in $\Kbar[[\e]]$ such that  
\begin{equation}\label{eq:satis}
\v(0, \dots, 0) = \balpha \quad \text{and} \quad p_1(\v, \e) = \cdots = p_n(\v, \e) = 0.
\end{equation}  

Furthermore, the vector of power series $\v$ can be truncated to an arbitrary degree $\delta$ using the ${\sf Lifting}$ algorithm, such that  
\[
\v^{(\lceil \log_2(\delta) \rceil)} = \v \mod \langle \e \rangle^\delta,
\]
where $\v^{(\lceil \log_2(\delta) \rceil)}$ is the output of the ${\sf Lifting}$ algorithm, which takes $\h$, $\balpha$, and $\delta$ as input. The complexity of computing this truncation is  
\[
\softO((nL + n^4){\mathcal{M}}(\delta, m))
\]
operations in $\K$, where $L$ is the length of a straight-line program computing $\h$.  
\end{proposition}

\begin{algorithm}[h] 	 
\caption{${\sf Lifting}(\h, \balpha, \delta)$}
\label{alg:lift}

\begin{flushleft}
{\bf Input:}
\begin{itemize}
  \item A sequence of polynomials $\h = (p_1, \dots, p_n)$ in $\K[\x, \e]^n$, where $\x = (x_1, \dots, x_n)$ and $\e = (u_1, \dots, u_m)$.
  \item A point $\balpha = (\alpha_1, \dots, \alpha_n) \in \Kbar{}^n$.
  \item A positive integer $\delta$.
\end{itemize}

{\bf Assumptions:}
\begin{itemize}
  \item $\balpha$ is a root of $\h(\x, 0, \dots, 0)$.
  \item The Jacobian matrix of $\h(\x, 0, \dots, 0)$ with respect to $\x$ has full rank at $\balpha$.
\end{itemize}

{\bf Output:} A vector that truncates $\balpha$ in $\K[[\e]]$ with precision $\delta$.
\end{flushleft}

\begin{enumerate}
\item Initialize $\v^{(0)} = \balpha$.

\item Compute the Jacobian matrix $\jac$ of $\h$ with respect to $\x$.

\item For $k = 1$ to $\lceil \log_2(\delta) \rceil$, do:
  \begin{enumerate}
    \item Update:
    \[
    \v^{(k)} = 
    \begin{pmatrix}
    v^{(k-1)}_1 \\ \vdots \\ v^{(k-1)}_n
    \end{pmatrix} 
    - \left(\jac(\v^{(k-1)}, \e)\right)^{-1} 
    \begin{pmatrix}
    p_1(\v^{(k-1)}, \e) \\ \vdots \\ p_n(\v^{(k-1)}, \e)
    \end{pmatrix}.
    \]
  \end{enumerate}
\item Return $\v^{(\lceil \log_2(\delta) \rceil)}$.
\end{enumerate}
\end{algorithm}

\begin{example}
Let us consider the polynomials $$\h = (p_1, p_2)= (x_1 + x_2 - u_1 - 2 , x_1^2 + x_2^2 - u_2 - 10)$$ in $\K[x_1, x_2, u_1, u_2]$. The point $(-1, 3)$ is a root of the system $$p_1(x_1, x_2, 0, 0) = p_2(x_1, x_2, 0, 0) = 0.$$ The Jacobian matrix of $(p_1(x_1, x_2, 0, 0), p_2(x_1, x_2, 0, 0))$ with respect to $(x_1, x_2)$ is  
\[
\begin{pmatrix} 
1 & 1 \\ 
2x_1 & 2x_2 
\end{pmatrix},
\]
which has full rank at $(-1, 3)$.

 The power series  
\[
v_1 = -1 + \frac{3}{4}u_1 - \frac{1}{8}u_2 + \frac{5}{64}u_1^2 - \frac{1}{64}u_1u_2 + \frac{1}{256}u_2^2 + \langle u_1, u_2 \rangle^2
\]
and  
\[
v_2 = 3 + \frac{1}{4}u_1 + \frac{1}{8}u_2 - \frac{5}{64}u_1^2 + \frac{1}{64}u_1u_2 - \frac{1}{256}u_2^2 + \langle u_1, u_2 \rangle^2
\]
in $\K[[u_1, u_2]]$ satisfy the following: $v_1(0, 0) = -1, \ v_2(0, 0) = 3, \ \text{and}$
\[
 \ p_1(v_1, v_2, u_1, u_2) = p_2(v_1, v_2, u_1, u_2) = 0.
\]
\end{example}

\subsection{The existence of power series solutions}
We first prove that the sequence $\left(\v^{(k)}\right)_{k \geq 0}$ is well-defined. To do so, we establish the following claims:

\begin{lemma}\label{lemma:convergent} Let $\jac$ denote the Jacobian matrix of $\h$ with respect to $\x$. Then, for any integer $k \in \mathbb{Z}_{\geq 0}$, 
\begin{itemize}
    \item[$(${\bf a}$)$] the determinant of $\jac$ at $\v^{(k)}$ is invertible in $\K[\e]/\langle \e \rangle$ and
    \item[$(${\bf b}$)$] $p_i(\v^{(k)}, \e) = 0 \mod \langle \e \rangle^{2^k}$ for all $i = 1, \dots, n$.
\end{itemize}
\end{lemma}

            \begin{proof}
We prove these claims by induction on $k$. Base Case ($k = 0$):
For $k = 0$, since $\v^{(0)} = \balpha$, the claims follow directly from the assumptions: 
 $\balpha$ is a root of $\h(\x, 0, \dots, 0)$ and  the Jacobian matrix of $\h(\x, 0, \dots, 0)$ with respect to $\x$ has full rank at $\balpha$.  

Assume that $(${\bf a}$)$ and $(${\bf b}$)$ hold for some $k \geq 0$. We will show that they also hold for $k+1$.  First, from the update rule,
\begin{equation} \label{eq:jac_update}
\v^{(k+1)} - \v^{(k)} = -\jac(\v^{(k)}, \e)^{-1}
\begin{pmatrix}
p_1(\v^{(k)}, \e) \\
\vdots \\
p_n(\v^{(k)}, \e)
\end{pmatrix}.
\end{equation}
Using the induction hypothesis for $k$, we can deduce that
\begin{equation} \label{eq:convergencu_step}
\v^{(k+1)} - \v^{(k)} = 0 \mod \langle \e \rangle^{2^k}.
\end{equation}

\paragraph*{Proof of $(${\bf a}$)$ for $k+1$:}
Let $J$ denote the determinant of $\jac$, the Jacobian matrix of $\h$ with respect to $\x$. Using the Taylor expansion:
\begin{multline*}
J(\v^{(k+1)}, \e) = J(\v^{(k)}, \e) + \sum_{i=1}^n \frac{\partial J}{\partial x_i}(\v^{(k)}, \e) (\v^{(k+1)} - \v^{(k)}) \\
\mod \langle \v^{(k+1)} - \v^{(k)} \rangle^2,
\end{multline*}
where $\langle \v^{(k+1)} - \v^{(k)} \rangle$ is the ideal in $\Kbar(\e)$ generated by the components of $\v^{(k+1)} - \v^{(k)}$. 

By the induction hypothesis for $k$, we know $J(\v^{(k)}, \e)$ is non-zero in $\K[\e]/\langle \e \rangle$. Additionally, from \eqref{eq:convergencu_step}, $\v^{(k+1)} - \v^{(k)} = 0 \mod \langle \e \rangle$. Therefore, $J(\v^{(k+1)}, \e) \neq 0 \mod \langle \e \rangle$, proving $(${\bf a}$)$ for $k+1$.

\paragraph*{Proof of $(${\bf b}$)$ for $k+1$:}
Multiplying both sides of \eqref{eq:jac_update} by the gradient of $p_i$ with respect to $\x$, we obtain:
\[
\left(\frac{\partial p_i}{\partial x_1}, \dots, \frac{\partial p_i}{\partial x_n} \right) 
\left(\v^{(k+1)} - \v^{(k)}\right)^T = -p_i(\v^{(k)}, \e).
\]

Using the Taylor expansion of $p_i$ between $\v^{(k+1)}$ and $\v^{(k)}$, we have
\begin{multline*}
p_i(\v^{(k+1)}, \e) = p_i(\v^{(k)}, \e) + \sum_{j=1}^n 
\frac{\partial p_i}{\partial x_j}(\v^{(k)}, \e)(\v^{(k+1)} - \v^{(k)}) \\
\mod \langle \v^{(k+1)} - \v^{(k)} \rangle^2.
\end{multline*}

Combining these two facts, we conclude:
\[
p_i(\v^{(k+1)}, \e) = 0 \mod \langle \v^{(k+1)} - \v^{(k)} \rangle.
\]
Finally, together with \eqref{eq:convergencu_step}, this implies that
\[
p_i(\v^{(k+1)}, \e) = 0 \mod \langle \e \rangle^{2^{k+1}}.
\]
Thus, $(${\bf b}$)$ holds for $k+1$.
\end{proof}

Let us conclude the existence of the vector of power series $\v$ in $\Kbar[[\e]]^n$ that satisfies~\eqref{eq:satis}. As shown in~\eqref{eq:convergencu_step}, the equation 
\[
\v^{(k+1)} - \v^{(k)} = 0 \mod \langle \e \rangle^{2^k}
\]
holds for any $k \in \ZZ_{\ge 0}$. Thus, for any $i = 1, \dots, n$, the sequence of functions $\left(v_i^{(k)}\right)_{k \in \ZZ_{\ge 0}}$ converges to a power series $v_i$ in $\K[[\e]]$. 

Furthermore, by Lemma~\ref{lemma:convergent}$(${\bf b}$)$, we have:
\[
p_i(\v^{(k)}, \e) = 0 \mod \langle \e \rangle^{2^k}, \quad \text{for } i = 1, \dots, n \text{ and all } k \in \ZZ_{\ge 0}.
\]
This implies that $p_i(\v, \e) = 0$ in $\Kbar[[\e]]$ for all $i = 1, \dots, n$. 

Finally, the relation in~\eqref{eq:convergencu_step} also ensures that:
\[
\v(0, \dots, 0) = \balpha.
\]
This completes the proof of the existence of the vector $\v$ of power series.


\subsection{The cost of computing the truncations}

We now analyze the complexity of computing $\v^{\lceil\log_2(\delta)\rceil}$ from $\balpha$ and $\h$, completing the proof of Proposition~\ref{prop:newton}.  

\begin{lemma} 
Let $L$ be the complexity of computing $\h = (p_1, \dots, p_n)$. The complexity to compute the approximation $\v^{\lceil\log_2(\delta)\rceil}$ of $\balpha$ in $\Kbar[[\e]]$ with precision $\delta$ is 
\[
O((nL + n^4){\mathcal{M}}(\delta, m))
\]
operations in $\K$. 
\end{lemma}

\begin{proof} 
The complexity of computing the first partial derivatives $\left(\frac{\partial p_i}{\partial x_j}\right)_{1 \le i,j \le n}$ of $\h$ is $O(nL)$ operations in $\K$, as established by \cite[Theorem~1]{Baur1983complexity} or \cite[Lemma~25]{giusti1997lower}. Consequently, computing the Jacobian matrix $\jac$ of $\h$ with respect to $\x$ also requires $O(nL)$ operations in $\K$. 

Next, consider the computation of $\v^{(k+1)}$ given $\v^{(k)}$. Evaluating the Jacobian matrix $\jac$ and the vector $\h$ at $\v^{(k)}$ requires $O(nL)$ operations in $\K[[u_1, \dots, u_m]] / \langle \e \rangle^{2^{k+1}}$, using Baur-Strassen's algorithm~\cite{Baur1983complexity}. Additionally,
 to compute $\v^{(k+1)}$, one also needs to invert the matrix $\jac(\v^{(k)}, \e)$, which costs $O(n^4)$ operations in the same ring. This matrix inversion can be performed, for example, using Leverrier’s algorithm~\cite{le1840variations}, which involves some divisions by intergers \(\{1, \dots, n\}\), or alternatively using Berkowitz’s algorithm~\cite{berkowitz1984computing}, which avoids divisions entirely.

Finally, the cost of a single operation in the quotient ring \\ $\K[[u_1, \dots, u_m]] / \langle \e \rangle^{2^{k+1}}$ is $\mathcal{M}(2^{k+1}, m)$ operations in $\K$.
Therefore, the total cost to compute $\v^{(k+1)}$ from $\v^{(k)}$ is 
\[
O((nL + n^4){\mathcal{M}}(2^{k+1}, m))
\]
operations in $\K$.

Summing up over all iterations $k$ up to $\lceil \log_2(\delta) \rceil$, the total complexity to compute the approximation $\v^{\lceil\log_2(\delta)\rceil}$ of $\balpha$ is 
\[
O((nL + n^4) \sum_{k=0}^{\lceil \log_2(\delta) \rceil} {\mathcal{M}}(2^{k+1}, m)) = O((nL + n^4){\mathcal{M}}(\delta, m)),
\]
where the equality follows from the monotonicity property of $\mathcal{M}$. 
\end{proof}

 
\section{The Main Result} \label{sec:main}

Let $\R = \mathbb{K}[x_1, \dots, x_n]$ be a multivariate polynomial ring over a field $\mathbb{K}$. Consider $n$ algebraically independent elements $g_1, \dots, g_n$ in $\R$. Let $\S$ denote the subring of $\R$ generated by $g_1, \dots, g_n$, and let $h$ be an element of $\S$. Then, there exists a unique polynomial $f \in \mathbb{K}[u_1, \dots, u_n]$, where $\e = (u_1, \dots, u_n)$ are new variables, such that  
\[
h = f(g_1, \dots, g_n).
\]  
In this section, we present an algorithm for computing the polynomial $f$ and analyze its complexity.

\subsection{Full rank matrices}

We will need the following result to establish the correctness of our main algorithm.  

\begin{lemma}  \label{lemma:ipen}
Let \(\K\) be a field of characteristic zero and we let $\h = (p_1, \dots, p_n)$ be a sequence of polynomials in $\K[\x]$, with $\x = (x_1, \dots, x_n)$. Then there exists a non-empty Zariski open set 
$\mathscr{U} \subset \Kbar{}^n$ such that
for any $\balpha \in \mathscr{U}$, the Jacobian matrix of $\Q$ with
respect to $\x$ has full rank at $\balpha$, where   
\[
\Q = (p_1 - p_1(\balpha), \dots, p_n - p_n(\balpha))
\]  is a sequence of polynomials in \(\K[\x].\)\label{lemma:shift}  
\end{lemma}
\begin{proof}  
Let $\z = (z_1, \dots, z_n)$ be a set of new variables. We set
$\mathfrak{Q} = (p_1-z_1,
\dots, p_n-z_n)$, a polynomial system in $\K[\x, \z]$. For a point
$\rho = (\rho_1, \dots, \rho_n) \in \Kbar{}^n$, we denote by
$\Theta_\rho$ the mapping  
\begin{align*}
  \Theta_\rho : \K[\z][x_1, \dots, x_n] &\rightarrow \Kbar[x_1,
                                           \dots, x_n], \\ 
   \mathfrak{Q}  &\mapsto (p_1-\rho_1, \dots, p_n-\rho_n)
\end{align*}
by setting $\z$ equal to $\rho$. Note that, for any $\rho
\in \Kbar{}^n$, the Jacobian matrix of $\Theta_\rho(\mathfrak{Q})$
with respect to $\x$ equals that of $\Q$. 

Consider the mapping 
\begin{align*}
 \Theta : (\balpha, \rho) \in \Kbar{}^n \times \Kbar{}^n &\rightarrow
                                                            \Theta_\rho(\mathfrak{Q})(\balpha).       
\end{align*} 
Since the columns corresponding to partial derivatives of
$\mathfrak{Q}$ with respect to $\z$ contain a $\diag(-1, \dots, -1)$
matrix, the Jacobian matrix of $\mathfrak{Q}$ has full rank at
all points $(\balpha, \rho)$ of its zero set. In other words, $\mathbf{0}$
is a regular value of $\Theta$. 

By Thom's weak transversality theorem (see
e.g., \cite[Proposition B.3]{din2017nearly} for the algebraic version),
there exists a non-empty Zariski open set $\mathscr{O} \subset
\Kbar{}^n$ such that for $\rho \in \mathscr{O}$, $\mathbf{0}$ is a 
regular value of the induced mapping
\begin{align*}
     \balpha \in \Kbar{}^n  \rightarrow
  \Theta_\rho(\mathfrak{Q})(\balpha). 
\end{align*} 
That is, the Jacobian matrix of
$\Theta_\rho(\mathfrak{Q})$ has rank $n$ at any root $\balpha \in
\Kbar{}^n$ of $\Theta_\rho(\mathfrak{Q})$.  

As a consequence, for the dense Zariski open set $\mathscr{U} :=
\h^{-1}(\mathscr{O})$, we have for all points $\balpha \in
\mathscr{U}$ that the Jacobian of $\h$ evaluated at $\balpha$
has full rank. In addition, $\balpha$ is a root of $\Q$, and the
Jacobian of $\h$ is the same as that of $\Q = \h -
\P(\balpha)$. Altogether, we obtain our claim.
\end{proof}
Let \(\K\) be a field of characteristic zero or sufficiently large characteristic. We remark that there is a well-known result (the {Jacobian criterion}) which provides a simple condition for the algebraic independence of \( n \) polynomials \( f_1, \dots, f_n \) in \( n \) variables. Precisely, these polynomials are algebraically independent if and only if the determinant of their Jacobian matrix is nonzero. Equivalently, the Jacobian matrix has full rank over the function field if and only if \( f_1, \dots, f_n \) are algebraically independent—under the assumption that \(\K\) has characteristic zero or a sufficiently large characteristic. For more details, we refer the reader to \cite{humphreys1992reflection} for the case of characteristic zero, and to \cite{beecken2013algebraic} for the case of large characteristic.

It is also worth noting that the purpose of Lemma~\ref{lemma:ipen} is not to prove the algebraic independence of the polynomials. Rather, it establishes that the Jacobian matrix of these polynomials has full rank at any common zero of the system. As illustrated in Example~\ref{ex:count} below, although the polynomials \( (x_1 + x_2,\, x_1^2 + x_2^2) \) (the Newton power sums) are algebraically independent, their Jacobian matrix does not have full rank at their common solutions.

We remark that the result in Lemma~\ref{lemma:ipen} still holds when \(\K\) is a finite field with a sufficiently large number of elements.


\subsection{Our main idea}
\label{subsec:idea}
Our main approach involves eliminating the variables $\x$ from the system  
\[
\h = (g_1 - u_1, \dots, g_n - u_n) \in \K[\x, \e]^n
\]  
using linear algebra techniques, specifically Hensel-Newton lifting.

Assume there exists a point $\balpha = (\alpha_1, \dots, \alpha_n) \in \Kbar{}^n$ such that  
$\balpha$ is a solution of $\h(\x, 0, \dots, 0)$ and the Jacobian  
matrix of $\h$ with respect to $\x$ has full rank at  
$\balpha$. Then, by Proposition~\ref{prop:newton}, there exists a  
unique vector of power series $\v = (v_1, \dots, v_n) \in \K[[\e]]^n$  
such that  
\[
\v_1(0, \dots, 0) = \alpha_1, \dots, v_n(0, \dots, 0) = \alpha_n \] {\rm and} \begin{equation*}p_1(\v, \e) =  
\cdots =p_n(\v, \e) = 0.  
\end{equation*}  
In order to find ${f}$, we only need  
truncations of $(v_1, \dots, v_n)$ at precision $\Delta$, where \(\Delta\) is the degree of \(f\), which can be  
done by using the ${\sf Lifting}$ algorithm. We then finally evaluate  
$f$ at these truncated power series to obtain ${f}$.

However, for a root $\balpha$ of $\h(\x, 0, \dots, 0)$, the  
Jacobian matrix of $\h(\x, 0, \dots, 0)$ with respect to $\x$ is  
not always full rank at $\balpha$.  
\begin{example}  \label{ex:count}
Let us consider $n=2$ and $(g_1, g_2) = (x_1+x_2, x_1^2+x_2^2)$. Then $\h= (x_1+x_2-u_1, x_1^2+x_2^2-u_2)$, and the Jacobian matrix of $\h$ with respect to $(x_1, x_2)$ is  
\[  
\jac = \begin{pmatrix} 1 & 1 \\ 2x_1 & 2x_2 \end{pmatrix}.  
\]  
The point $\balpha = (0,0)$ is a solution of $\h(x_1, x_2, 0,0)$; however, the rank of $\jac$ at $\balpha$ is 1. \label{ex:not_full}  
\end{example}

In order to deal with this situation, we take a point $\balpha \in \mathscr{U}$, where $\mathscr{U}$ is the non-zempty Zariski open set defined in Lemma  \ref{lemma:shift} (that is  a generic
point in \(\Kbar{}^n\)). Then we compute a polynomial system  
\[
\Q := (g_1-u_1-g_1(\balpha), \dots, g_n - u_n-g_n(\balpha)) \in \K[\x, \e].
\] In this way, the point $\balpha$ is a root of $\Q(\x, 0, \dots, 0)$, and the Jacobian matrix of $\Q(\x, 0, \dots, 0)$ with respect to $\x$ is full rank at $\balpha$ by Lemma~\ref{lemma:shift}.  

\begin{example}  
Let us continue Example~\ref{ex:not_full}. We take a generic point $\balpha = (-1, 3)$ in $\Kbar{}^2$. Then $\Q = (x_1+x_2-u_1-2, x_1^2+x_2^2-u_2-10)$. The Jacobian matrix  
\[  
\begin{pmatrix} 1 & 1 \\ 2x_1 & 2x_2 \end{pmatrix}  
\]  
with respect to $(x_1, x_2)$ of $\Q$ has full rank $2$ at $\balpha$.  
\end{example}

Then, if $\v^{\lceil \log_2(d)\rceil}$ is the truncation of the vector of power series $\v$, and the polynomial $\ell$—which equals ${f}(u_1 - g_1(\balpha), \dots, u_n - g_n(\balpha))$—is the evaluation of $h$ at $\v^{\lceil \log_2(d)\rceil}$, we apply the translation  
\[
(u_i)_{1 \le i \le n} \leftarrow (u_i + g_i(\balpha))_{1 \le i \le n}
\]  
to obtain the polynomial ${f}$.


\subsection{The main algorithm (proof of Theorem \ref{main:thm})}
Our main result in this paper is as follows. For convenience, we restate the main theorem here. Recall that the field \(\K\) is of characteristic zero.
\begin{theorem} \label{thm:main_2}
 Then there exists a randomzied  algorithm algorithm ${\sf Representation}$ which, given $(g_1, \dots, g_n)$ and $f$, returns $f$ in $\K[u_1, \dots, u_n]$ such that  
$f(g_1, \dots, g_n) = h$,
using $${\softO\big((nL_1 + n^4+L_2){\mathcal{M}}(\Delta, n)}\big)$$  operations in $\K$, where $L_1$ and $L_2$ are the lengths of straight-line programs computing $(g_1, \dots, g_n)$ and $h$ respectively, and $\Delta$ is a bound on the degree of $g$. 
\end{theorem}

\begin{algorithm}[h] 	 
  \caption{${\sf Representation}\left((g_1, \dots, g_n), h, \Delta \right)$} 
  \label{alg:representation}

  \begin{flushleft}
  {\bf Input:}
  \begin{itemize}
    \item $n$ algebraically independent polynomials $g_1, \dots, g_n$ in $\R$
    \item a polynomial $h \in \K[g_1, \dots, g_n]$
    \item a positive integer \(\Delta\) as the degree bound for the output
  \end{itemize}

  {\bf Output:} The polynomial $f \in \K[u_1, \dots, u_n]$ such that
    \[
    h = f(g_1, \dots, g_n)
    \]
  \end{flushleft}

  \begin{enumerate}
    \item Take a generic point $\balpha = (\alpha_1, \dots, \alpha_n) \in   \Kbar{}^n$. \label{step:lstart} 
    
    \item Define the polynomials $\Q = (g_1 - u_1 - g_1(\balpha), \dots, g_n - u_n - g_n(\balpha))$ in $\K[x_1, \dots, x_n, u_1, \dots, u_n]$.
    \item Compute $(v_1, \dots, v_n) = {\sf Lifting}(\Q, \balpha, \Delta) \in \K[[u_1, \dots, u_n]]^n$. \label{step:l}
    \item Compute $h(v_1, \dots, v_n)$ and truncate the result at degree $\Delta$ to obtain a polynomial $\ell$. \label{step:e}
    \item Return $\ell(u_1 + g_1(\balpha), \dots, u_n + g_n(\balpha))$. \label{setp:f}
  \end{enumerate}
\end{algorithm}

\begin{proof}
The correctness of the algorithm follows from Subsection~\ref{subsec:idea}. It remains to establish a complexity analysis of the algorithm.

We first evaluate \(g_1, \dots, g_n\) at \(\balpha\) using a straight-line program (SLP) of length \(L_1\), which requires \(O(L_1)\) operations in \(\K\). At the core of the algorithm, specifically at Step~\ref{step:l}, we compute the truncated power series \((v_1, \dots, v_n)\) using the \textsf{Lifting} algorithm from Proposition~\ref{prop:newton}. This step takes  
\[
\softO((nL_1 + n^4)\,{\mathcal{M}}(\Delta, n))
\]  
operations in \(\K\).

Next, we evaluate \(h\) at these truncated power series in Step~\ref{step:e}. This requires \(O(L_2)\) operations on \(n\)-variate power series truncated at total degree \(\Delta\), yielding a total cost of  
\(
O\big(L_2\,\mathcal{M}(\Delta, n)\big)
\)  
operations in \(\K\). This step produces the polynomial  
\[
\ell := f(u_1 - g_1(\balpha), \dots, u_n - g_n(\balpha)).
\]  
Moreover, a straight-line program of length \(L_1 + n + L_2\) can be used to compute \(\ell\): we simply add \(n\) additional nodes to implement the translation from \(u_i\) to \(u_i + g_i(\balpha)\), and as discussed above, the values \(\{g_i(\balpha)\}_{i=1}^n\) can be computed via an SLP of length \(L_1\).

We then apply the translation \((u_i)_{1 \le i \le n} \leftarrow (u_i + g_i(\balpha))_{1 \le i \le n}\) to \(\ell\) in order to recover \(f\). To perform this translation, we modify the operations in the SLP representing \(\ell\), replacing each occurrence of \(u_i\) with \(u_i + g_i(\balpha)\). Since each instruction involving a variable can be updated in constant time, and the SLP for \(\ell\) has length \(L_1 + n + L_2\), this step requires  
\[
O(L_1 + n + L_2)
\]  
operations in \(\K\). Once the translation is applied, we evaluate the resulting SLP to obtain the final polynomial \(f\). Since each instruction is evaluated once, this again takes \(O(L_1 + n + L_2)\) operations, which is negligible compared to the preceding steps.

Thus, the total cost of computing the polynomial \(f\) is  
\[
\softO\big((nL_1 + n^4 + L_2)\,{\mathcal{M}}(\Delta, n)\big)
\]  
operations in \(\K\), as claimed.
\end{proof}

Note that our algorithm is randomized (or probabilistic) in the sense that it makes random choices of points (i.e., \(\balpha\)) which lead to correct computations; these points lie outside certain Zariski closed subsets of suitable affine spaces (see Lemma~\ref{lemma:ipen}). In this sense, our algorithm is of the Monte Carlo type (see, e.g., \cite{Gat03}), meaning that it returns the correct output with high probability---at least a fixed value greater than \(1/2\). The error probability can be made arbitrarily small by repeating the algorithm multiple times.

Although we do not estimate the exact probability that our algorithm is correct, it can be controlled using the Schwartz--Zippel lemma (see, e.g., \cite{schwartz1980fast} or \cite[Lemma~6.44]{Gat03}). More precisely, we define sufficiently large finite subsets of integers whose cardinalities depend on the degrees of the polynomials defining the aforementioned Zariski closed sets, and we then make the necessary random choices uniformly from these sets.

Finally, it is worth noting that although the input of Algorithm~\ref{alg:representation} requires a positive integer \(\Delta\) as a degree bound for the output polynomial \(f\), we can perform Steps~\ref{step:lstart} to~\ref{step:l} using the {\sf Lifting} procedure with successively doubling integers, such as \(2, 4, 8, \dots\). After each {\sf Lifting} step, we can check whether the polynomial obtained at Step~\ref{setp:f} matches the desired output \(f\). It is not difficult to verify that this adaptive approach preserves the overall complexity stated in Theorem~\ref{thm:main_2}.

\section{Some  degree bounds}

\label{subsec:invariant}
Since the complexity of the main algorithm depends on the degree of the polynomial $f$, this final section is devoted to providing bounds on the degree of $f$ in terms of $h$ and $g_1, \dots, g_n$.

\subsubsection*{\bf Weighted Homogeneous Polynomials}  

Weighted degrees generalize the standard notion of polynomial degree and naturally appear in many domains, especially when dealing with polynomial invariants under group actions. They help track transformation properties and have been extensively studied in various contexts. For example, weighted degrees play a crucial role in Gröbner bases for weighted homogeneous systems \cite{FSV15} and in weighted B\'ezout’s theorem \cite{damon1997global}, which provides bounds on the number of isolated solutions. Additionally, they have been used to speed up the complexity of solving certain polynomial systems \cite{FLSSV2021, labahn2021homotopy}.

Given a tuple of \( n \) positive integers \( W = (w_1, \dots, w_n) \in \mathbb{N}^n_{>0} \), the \textit{\( W \)-degree} of a monomial \( x_1^{k_1} \cdots x_n^{k_n} \) is defined as  
\[
\deg_W(x_1^{k_1} \cdots x_n^{k_n}) = \sum_{i=1}^{n} w_i k_i.
\]  
A polynomial is called \textit{\( W \)-homogeneous} if all its monomials have the same \( W \)-degree. In other words, a polynomial \( g \) is \( W \)-homogeneous of \( W \)-degree \( d \) if and only if  
\[
\mathrm{hom}_W(g) := g(x_1^{w_1}, \dots, x_n^{w_n})
\]  
is a homogeneous polynomial of degree \( d \).

\begin{example} \label{ex_weight}
Let \( W = (5,2) \). The polynomials \( g_1 = x_1^2 + x_2^5 \) and \( g_2 = x_2 \) are \( W \)-homogeneous of \( W \)-degrees \( 10 \) and \( 2 \), respectively.  
\end{example}

\begin{lemma} \label{lemma:weight}
    Let \(W = (w_1, \dots, w_n) \in \mathbb{N}^n_{>0}\) and
     \( g_1, \dots, g_n \) be a sequence of \( W \)-homogeneous, algebraically independent polynomials in \( \mathbb{K}[x_1, \dots, x_n] \). Let \( \mathcal{S} = \mathbb{K}[g_1, \dots, g_n] \), and suppose that \( h \in \mathcal{S} \). 

    Let \( f \) be the unique polynomial in \( \mathbb{K}[u_1, \dots, u_n] \) such that
    \[
    h = f(g_1, \dots, g_n).
    \]
    Then,  
    \[
    \deg_{W'}(f) = \deg_W(h),
    \]
    where \( W' = (\deg_W(g_1), \dots, \deg_W(g_n)) \).
\end{lemma}
\begin{proof}
    Let  \( W' = (w_1', \dots, w_n') \) and \[ f = \sum_{\a = (a_1, \dots, a_n)} \, c_\a\, u_1^{a_1} \cdots u_n^{a_n}. \]Since \( h = f(g_1, \dots, g_n) \), we have 
    \[
    h = \sum_{\a} \, c_\a\, g_1^{a_1} \cdots g_n^{a_n}.
    \]
    Since \( g_i \)'s are \( W \)-homogeneous, the polynomial \( h \) can be written as
    \[
    h = h_{\rm high} + h_{\rm low},
    \]
    where \( h_{\rm high} = \sum_{\a'} \, c_{\a'}\, g_1^{a'_1} \cdots g_n^{a'_n} \) is the sum of all terms (in \( g_i \)'s) in \( h \) such that the degree of each term \( g_1^{a_1} \cdots g_n^{a_n} \) reaches the maximum \( W \)-degree among all terms in \( h \). In other words, \( \deg_W(g_1^{a'_1} \cdots g_n^{a'_n}) \) is maximized for all \( \a' = (a_1', \dots, a_n') \), that is,
    \[
    \deg_W(g_1^{a'_1} \cdots g_n^{a'_n})  = \max_{ (a_1, \dots, a_n)} \left( a_1 \deg_W(g_1) + \cdots + a_n \deg_W(g_n) \right).
    \]
    
    Moreover, since the \( g_i \)'s are algebraically independent, \( h_{\rm high} \) is a non-zero polynomial. Therefore, we have:
    \[
    \deg_W(h) = \deg_W(h_{\rm high}) \ {\rm and} \  \deg_W(h_{\rm high}) =   \deg_W\left(g_1^{a'_1} \cdots g_n^{a'_n}\right).
    \]    
    Thus, the degree of \( h \) satisfies:
    \[
    \deg_W(h) = \max_{(a_1, \dots, a_n)} \left( a_1 \deg_W(g_1) + \cdots + a_n \deg_W(g_n) \right),
    \]
    which simplifies to:
    \[
    \max_{(a_1, \dots, a_n)} \left( a_1 w_1' + \cdots + a_n w_n' \right),
    \]
    which by definition equals \( \deg_{W'}(f) \).
\end{proof}

\begin{example}
    Consider \( W \), \( g_1 \), and \( g_2 \) as in Example \ref{ex_weight}, and let \( h = x_1^2 \) (so that \( \deg_W(h) = 10 \)). Then \(W' = (10, 2)\) and \( f = u_1 - u_2^5 \) has \( \deg_{W'}(f) = 10 \), which implies that \(\deg_{W'}(f) = \deg_W(h)\). 
\end{example}

\subsubsection*{\bf Polynomials Invariant Under Finite  Groups} Polynomials invariant under finite groups play a crucial role in many areas of mathematics and applications in science and engineering.  We now discuss a degree bound for $f$ when $h$ is invariant under the action of a finite group.

Let \( GL(n, \mathbb{K}) \) be the set of all invertible \( n \times n \) matrices with entries in \( \mathbb{K} \), and let \( \calG \subset GL(n, \mathbb{K}) \) be a finite matrix group. A polynomial \( h \in \mathbb{K}[x_1, \dots, x_n] \) is called \(\calG\)-$invariant$ if
\[
h = h(A \cdot \mathbf{x}) \quad \text{for all } A \in \calG.
\]
The set of all \(\calG\)-invariant polynomials is denoted by \( \mathbb{K}[x_1, \dots, x_n]^\calG \). It is known that \( \mathbb{K}[x_1, \dots, x_n]^\calG \) is finitely generated by homogeneous polynomials (see, for example, \cite[Theorem 5, Section 3, Chapter 7]{CLO07}). 

A well-known example is when \( \calG = S_n \), the symmetric group of permutation matrices. In this case, the ring of invariant polynomials \( \mathbb{K}[x_1, \dots, x_n]^{S_n} \) consists of all symmetric polynomials. This subring  of \( \K[x_1, \dots, x_n] \) is generated by elementary symmetric polynomials, power sum symmetric polynomials (over any field of characteristic 0), Schur polynomials, and complete homogeneous symmetric polynomials.

If \( \calG \) is a finite subgroup of \( GL(n, \mathbb{K}) \) and \( A \in \calG \), then \( A \) is called a \textit{pseudo-reflection} if it has an eigenvalue 1 of multiplicity \( n-1 \) and another eigenvalue of multiplicity 1.

\begin{theorem}[\cite{shephard1954finite, chevalley1955invariants, serre1965alg}]
     There exist \( n \) algebraically independent homogeneous invariants \( g_1, \dots, g_n \) such that \( \mathbb{K}[x_1, \dots, x_n]^\calG = \mathbb{K}[g_1, \dots, g_n] \) if and only if \( \calG \) is generated by pseudo-reflections. 
\end{theorem}

\noindent This theorem, combined with Lemma \ref{lemma:weight} for \( W = (1, \dots, 1) \), provides the following degree bound.

\begin{lemma}
    Let \( h \in \mathbb{K}[x_1, \dots, x_n]^\calG = \mathbb{K}[g_1, \dots, g_n] \), and let \( f \)  be the polynomial in \(\mathbb{K}[u_1, \dots, u_n] \) be such that 
    \(
    f(g_1, \dots, g_n) = h.
    \)
    Let \( W' = (\deg(g_1), \dots, \deg(g_n)) \). Then, the degree of \( f \) satisfies
    \[
    \deg_{W'}(f) = \deg(h).
    \]
    As a consequence, we have the bound
    \(
    \deg(f) \leq \deg(h).
    \)
\end{lemma}

\begin{example}
  Let us consider the polynomial \(h\) and the polynomials \( g_i \) from Section \ref{sec:ex}. The weight vector is given by  
  \[
  W' = (\deg(g_1), \deg(g_2), \deg(g_3)) = (1, 2, 3),
  \]
  and we have \( \deg(h) = 3 \). We have previously computed that  
  \[
  f = -\frac{1}{3}u_1^3 + u_1u_2 - u_1 + \frac{1}{3}u_3.
  \]
  Then, the weighted degree of \( f \) is  
  \(
  \deg_{W'}(f) = 3,
  \)
  which equals \( \deg(h) \). In addition,  \( \deg(f) = 3 \) (the standard degree, without weights), which is also equal to \( \deg_{W'}(f) \).
\end{example}

\begin{example}
  Let us consider \( h = x_1x_2x_3 \in \mathbb{K}[x_1, x_2, x_3]^{\mathcal{S}_3} \), a symmetric polynomial, and let 
  \[
  g_1 = x_1 + x_2 + x_3, \quad g_2 = x_1x_2 + x_2x_3 + x_1x_3, \quad g_3 = x_1x_2x_3.
  \]
  Then \( f = u_3 \) and \( (g_1, g_2, g_3) \) are \(W'\)-homogeneous polynomials with \( W' = (1, 2, 3) \). We have \( \deg_{W'}(f) = 3 = \deg(h) \), while \( \deg(f) = 1 \) (the standard degree, without weights).
\end{example}

Together with Theorem \ref{thm:main_2}, we have the following result regarding the complexity of computing \( f \) when \( h \) is a polynomial invariant under the action of a pseudo-reflection group.

\begin{theorem} \label{thm:spec}
Let \(\K\) be a field of characteristic zero and    let \(\calG\) be a pseudo-reflection subgroup of \( GL(n, \K) \), and let \( g_1, \dots, g_n \) be algebraically independent homogeneous polynomials such that 
    \[
    \mathbb{K}[x_1, \dots, x_n]^\calG = \mathbb{K}[g_1, \dots, g_n].
    \] 
    
    Then, for any \( h \in \mathbb{K}[x_1, \dots, x_n]^\calG \), one can compute the polynomial \( f \in \K[u_1, \dots, u_n] \) such that \( h = f(g_1, \dots, g_n) \) in  
    \[
    \softO\big((nL_1 + n^4 + L_2) {\mathcal{M}}(d, n) \big)
    \]
    operations in \( \K \), where \( L_1 \) and \( L_2 \) are the lengths of straight-line programs computing \( (g_1, \dots, g_n) \) and \( h \), respectively, and \( d = \deg(h) \).
\end{theorem}

\section{Conclusions and Topics for Future Research}

In this paper, we presented an algorithm along with its complexity analysis for computing the representation \( f \) of a polynomial \( h \) in a subring generated by algebraically independent polynomials \( g_1, \dots, g_n \). The complexity of our algorithm is linear in the size of the input polynomials \( h \) and \( g_i \)'s, and polynomial in \( n \) when the degree of \( f \) is fixed. An important special case arises when \( h \) is a polynomial invariant under the action of a finite pseudo-reflection group (e.g., symmetric groups \( S_n \); Weyl groups of Lie algebras such as dihedral groups \( D_n \), signed symmetric groups \( B_n \), etc.). 

To the best of our knowledge, previous work has addressed only the case where \( h \) is a symmetric polynomial and the \( g_i \)'s are elementary symmetric polynomials or homogeneous symmetric, or power symmetric polynomials. It is worth noting that, in the case of symmetric polynomials, the algorithm of Bl{\"a}ser and Jindal \cite{BlaserJindal18} is faster than ours. However, our results apply to a significantly more general setting.

Many questions remain open. To recall, our main approach is based on Newton-Hensel lifting. The key technical result involves shifting and the use of a generic starting point (Lemma~\ref{lemma:ipen}) to ensure that a certain Jacobian matrix has full rank—an essential condition for the lifting process. This step represents the main computational cost in our algorithm. Therefore, one direction for future research is to investigate whether, under certain settings, faster algorithms can be developed. An example is the case of symmetric polynomials.

Furthermore, we also want to compute the representation in equation \eqref{eq:second-first} for any polynomial \( h \) that is invariant under the action of a finite matrix group \( \calG \). As mentioned in the introduction, previous work by Dahan, Schost, and Wu \cite{dahan2009evaluation} provides an approach, but it remains computationally expensive.

The problem of decomposing multivariate polynomials is one of the main motivations for this paper. The result of this paper  is considered as one of two key steps in developing an algorithm to decompose certain polynomials. Therefore, designing a complete algorithm for decomposing multivariate polynomials is a focus of our furture research.

Finally, in the context of polynomial system solving, while the output polynomial \( f \) has more structure, solving new polynomial systems in equation \eqref{eq:decom} could be more computationally expensive. Therefore, a detailed analysis of this approach is the subject of future work. Additionally, providing tight bounds on the degree of \( f \) remains an open problem. Using weighted degrees to target this purpose, as well as employing a weighted version of \( \mathcal{M}(\Delta, n) \) to improve complexity estimates, are potential directions for future research. 
Furthermore, as mentioned in the introduction, straight-line program (SLP) representation is utilized in polynomial system-solving algorithms using Newton-Hensel lifting techniques. Therefore, investigating an upper bound for an SLP representing \(f\) is one of our future research goals.
\section{Acknowledgements} The author thanks the anonymous referee for their valuable comments and suggestions for future work.

\balance

\bibliographystyle{ACM-Reference-Format}
\bibliography{main.bib}

\appendix
\section{An example} \label{sec:ex}
We illustrate the steps of our main algorithm (Algorithm~\ref{alg:representation}) with a concrete example.  

Let us consider
\[
(g_1, g_2, g_3) = (x_1 + x_2 + x_3,\ 
x_1^2 + x_2^2 + x_3^2,\ 
x_1^3 + x_2^3 + x_3^3).
\]
Then \(\mathcal{S} = \K[x_1, x_2, x_3]^{S_3}\), where \(S_3\) is the symmetric group on three elements. This is the subring of \(\K[x_1, x_2, x_3]\) consisting of all symmetric polynomials in three variables. We consider the polynomial
\[
h = x_1^3 + x_2^3 + x_3^3 - 2x_1x_2x_3 - x_1 - x_2 - x_3,
\]
which lies in \(\mathcal{S}\) and has total degree \(d = 3\).  Here we also know that the degree of the output polynomial $\ell$ is bounded by the degree of the polynomial $h$.

The \texttt{Representation} algorithm proceeds as follows.
  First, we take a generic point
  $\balpha = (\alpha_1, \alpha_2, \alpha_3) = (4,6,0)$, then
   $$(c_1, c_2, c_3) = (g_1(\balpha),g_2(\balpha),
  g_3(\balpha)) =  (10, 52, 280).$$ 
  Then the polynomials $\Q$ are   
  \[
  q_1 = x_1+x_2+x_3 - u_1- 10, \quad q_2 = x_1^2+x_2^2+x_3^2 - u_2 - 52,
  \]
  and 
  \[
  q_3 = x_1^3+x_2^3+x_3^3 - u_3 - 280.
  \]

  Using the ${\sf Lifting}$ procedure with the input $\Q$,
  $(4,6,0)$, and $3$ as the 
  input, we obtain 
    \begin{multline*}   
    v_1 = \frac{-1}{512}u_1u_2u_3 - \frac{11}{96}u_1^3
  +\frac{33}{4096}u_2^3-\frac{1}{55296}u_3^3+ \cdots  \\ + \frac{3}{8} u_2
        - \frac{1}{24} u_3+ 4,
    \end{multline*}
        \begin{multline*}   
    v_2 =  \frac{11}{2592}u_1u_2u_3 + \frac{1}{24}u_1^3
  -\frac{197}{31104}u_2^3+\frac{29}{1679616}u_3^3 + \cdots  \\
      - \frac{1}{6} u_2 + \frac{1}{36} u_3
      + 6, {\rm ~and}
    \end{multline*}
        \begin{multline*}   
    v_3 =\frac{-95}{41472}u_1u_2u_3 +\frac{7}{96}u_1^3
  -\frac{1715}{995328}u_2^3 +  \frac{11}{13436928}u_3^3+ \cdots  \\ + u_1
      - \frac{5}{24}u_2 + \frac{1}{72}u_3.
    \end{multline*}
Then we substitute $(x_1, x_2, x_3) = (v_1, v_2, v_3)$
into $h$ and truncate the result at degree $3$ to obtain 
\[
\ell = -\frac{1}{3}u_1^3 - 10u_1^2 + u_1u_2 - 49u_1 + 10u_2 + \frac{1}{3}u_3 + 270;
\]
then finally 
\[
{f} =  \ell(u_1-10, u_2-52, u_3-280)
= -\frac{1}{3}u_1^3  + u_1u_2 - u_1 + \frac{1}{3}u_3.
\]

\end{document}